\documentclass[%
 reprint,
 amsmath,amssymb,
 aps, showkeys,
]{revtex4-2}

\usepackage{graphicx}
\usepackage{dcolumn}
\usepackage{bm}
\usepackage{float}

\usepackage{natbib}
\usepackage{appendix}
\usepackage{amsthm}

\newtheorem{mytheorem}{Theorem}

\newtheorem{mycorollary}{Corollary}
\newtheorem{myproposition}{Proposition}

\begin{document}

\title{Dropping mortality by increasing connectivity in plant epidemics}

\author{Ignacio Taguas}
\email{i.taguas@upm.es}
\author{Jos\'e A. Capit\'an}%
\altaffiliation{Complex Systems Group}
\email{ja.capitan@upm.es}
\author{Juan C. Nu\~no}%
\email{juancarlos.nuno@upm.es}
\affiliation{%
Department of Applied Mathematics. Universidad Polit\'ecnica de Madrid.\\
Av. Juan de Herrera, 6, E-28040, Madrid.
}%

\date{\today}

\begin{abstract}
Pathogen introduction in plant communities can cause serious impact and biodiversity losses that may take long time to manage and restore. Effective control of epidemic spreading in the wild is a problem of paramount importance, because of its implications in conservation and potential economic losses. Understanding the mechanisms that hinder pathogen's propagation is, therefore, crucial. Usual modelization approaches in epidemic spreading are based in compartmentalized models, without keeping track of pathogen's concentrations during spreading. In this contribution we present and fully analyze a dynamical model for plant epidemic spreading based on pathogen's abundances. The model, which is defined on top of network substrates, is amenable to a deep mathematical analysis in the absence of a limit in the amount of pathogen a plant can tolerate before dying. In the presence of such death threshold, we observe that the fraction of dead plants peaks at intermediate values of network's connectivity, and mortality decreases for large average degrees. We discuss the implications of our results as mechanisms to halt infection propagation.
\end{abstract}


\maketitle


\section{Introduction}

Introductions of new plant pathogens into previously uncolonized areas is a major problem, since the feasible lack of defenses of the individuals in the area might cause fatal losses \citep{Cunniffe2016}. One example is the sudden oak death, caused by the broad host range oomycete \textit{Phytophthora ramorum}, which has caused devastating impacts on some North American and European forests \citep{Grunwald2019}. Especially, it has killed millions of oak and tanoak in California since its first detection in 1995. Another noteworthy example of the impact of plant pathogen introduction is the massive economic damage caused by \textit{Xylella fastidiosa}, a bacterium that affects 563 plant species from the Americas, Europe, the Middle East and Asia. For example, an exhaustive study of the impact of this pathogen in olive trees can be found in reference \citep{Schneider2020}.

During an epidemic outbreak, control decisions must be taken to lessen pathogen impact \citep{Anderson2004}. However, the restricted amount of time in which actions must be taken and the lack of information early in the epidemic often make the choice of action difficult \citep{Ferguson2001_2, Ferguson2001_1} and the consequences can be extremely detrimental. The previously mentioned sudden oak death is a clear example: as pointed out by \citeauthor{Cunniffe2016} \citep{Cunniffe2016}, insufficient measures taken to eradicate the disease in California has led to a point where statewide action to even slow the spread of \textit{P. ramorum} is no longer feasible; the pathogen has spread far enough that the only possible solution at present is local containment \citep{Cunniffe2016}.

Thus, developing tools that allow us to better understand the behavior of new epidemics, and that help us predict what the outcome of different control measures might be, is a question of paramount importance. The main aim of this work is the analysis of the dynamics and stability, under different conditions, of a model for the spread of infections in wild plant populations, so that it can later be applied to real epidemics.

Most epidemiological models are based on a compartmentalization of individuals according to their disease status~\citep{may1979,hens2019spatiotemporal}. This is a notable simplification, for many details of the epidemic are neglected, for example, differences in response between individuals. The study of infection propagation on networks provides a way to include in the simulations different parameters related to the shape of the field in which the epidemic occurs, such as the spatial location of the individuals and the connectivity between different individuals. It also allows a better mechanistic understanding of the spread of the epidemic \citep{Estrada2016}. Traditionally, networks have been successfully used in human and animal epidemiology; however, not much work has been done for plant epidemics \citep{keeling2005}.

Compartmentalized models like SIS, SIR, and further extensions of these, generally do not track pathogen concentrations over time. Instead, those models provide temporal variation of the fraction of individuals belonging to each class (infected, susceptible, recovered, etc.) during infection propagation over networked substrates~\citep{anderson1992infectious,keeling2005}. Clustering individuals, e.g. plants or trees, into different classes makes infection dynamics more tractable. Instead, in this contribution we focus on pathogen densities across infected individuals and their variation over time, and we introduce and fully analyze a dynamical model for infection spreading defined by density independent, per-capita fluxes of pathogen between individuals. In addition, we extract meaningful information when individuals are classified according to their response to infection. In particular, we assume that an individual is dead when pathogen's abundance increases above a certain threshold. This means that the individual is effectively removed from the network, and does not contribute to pathogen's propagation to other neighbouring plants. Individuals can tolerate moderate epidemic charges (below the death threshold). It is precisely the presence of this threshold what enriches infection dynamics ---which otherwise is amenable to analytical treatment in the absence of such a death threshold.

The effect of the substrate (mean field, network, lattice, etc.) on which epidemics take place has been extensively studied~\citep{keeling2005implications,castellano2010thresholds,cuesta2011struggle,capitan2011severe,stegehuis2016epidemic,de2016physics}. However, the influence of network average connectivity on the fractions of healthy, infected or dead individuals has been overlooked, surprisingly. Our main result is related to the number of dead individuals in the presence of a finite death threshold. We find that the fraction of dead plants peaks at intermediate network connectivity: for small connectivity, networks are basically disconnected and infection comes to a halt with a small number of dead plants. When networks are connected, increasing average degree favors channeling the pathogen across the network, the larger the mean connectivity the smaller the fraction of dead individuals. Therefore, an effective mechanism to lessen the epidemiological impact in plant communities can be the facilitation of pathogen spreading over the network ---for instance, by planting new pathogen's hosts. Such mechanism can be regarded as alternative to obstructing or limiting propagation by isolating infected individuals.

\section{Epidemic spreading dynamics}

Our approach is based on a deterministic epidemic dynamics, that unfolds on top of a network substrate, for which model parameters are random variables drawn from specified distributions. Across the network, each node represents an individual plant. Consider a plant pathogen infecting a network formed by $n$ plants or trees belonging to the same species. Let $x_i$ be the density of pathogen in plant individual $i$. When isolated (i.e., pathogen is not transported among individuals) we assume a logistic growth with rate $r>0$ for pathogen abundance in each plant. This implicitly means that there is a carrying capacity $K$ within each individual that limits pathogen abundance.

Plant individuals are accessible to host pathogen particles coming from adjacent plants. Pathogen's transport is bi-directional: infected individuals can release pathogen's particles to neighbors, and can also receive additional particles from neighboring plants. Such flux of pathogen occurs among connected individuals $i$ and $j$, for $i\ne j$. Networks are defined in a way such that two plants are connected if there is a non-zero probability of contagion between them. This probability depends on the dispersal ability of the pathogen, which is usually correlated to the distance between plant pairs. Therefore, we can assume that plants are disconnected if this probability is small ---in other words, only pairs of individuals with probability of contagion above a threshold will be connected.

Consider a focal plant $i$. If there is a link in the network between $i$ and $j$, the per-capita (per unit of pathogen abundance of the source individual) in-going flux of pathogen from individual $j$ to $i$ is denoted as $a_{ij}>0$, and the per-capita out-going flux of pathogen from $i$ to $j$ as $b_{ji}>0$. Otherwise, if $i$ and $j$ are disconnected, then $a_{ij}=a_{ji}=0$ as well as $b_{ij}=b_{ji}=0$. In addition, we set $a_{ii}=b_{ii}=0$ for $i=1,\dots,n$. Then, the amount of pathogen transported from $j$ to $i$ is equal to $a_{ij}x_j$, and the abundance transferred from $i$ to $j$ equals to $b_{ji}x_i$. Consequently, the coupled dynamics of pathogen concentrations are driven by a system of $n$ coupled differential equations,
\begin{equation}\label{eq:dyn}
\frac{dx_i}{dt} = rx_i\left(1 - \frac{x_i}{K}\right) + \sum_{j=1}^n a_{ij}x_j - x_i \sum_{j=1}^n b_{ji},
\end{equation}
for $i=1,\dots,n$. Observe that sums run over the set of neighbors of node $i$, given the restrictions imposed above for $a_{ij}$ and $b_{ij}$. The connectivity of the plant community defines two matrices, $A = (a_{ij})$ and $B = (b_{ij})$. These per-capita in- and out-going fluxes will be considered as random variables. The parameters that define pathogen's growth in isolation ($r$ and $K$) will be varied throughout this study.
Graphs are non-directed, meaning that if there is a flux of pathogen from individual $j$ to individual $i$, there is also pathogen transfer from $i$ to $j$. These fluxes, however, do not have to be balanced (this depends on the values of $a_{ij}$ and $b_{ji}$).

Observe that, if no restrictions are imposed to matrices $A$ and $B$, the amount of pathogen transferred from a node is not necessarily equal to the overall pathogen's amount that all its neighbors receive. In that case, the system is open and the overall amount of pathogen in the plant community can increase (in the presence of an external source) or decrease (if there are additional mechanisms that channel pathogen particles out of the system). Otherwise, if the amount of pathogen that is transported out of node $i$ is exactly equal to the overall amount that all of its neighbors receive, then there are no losses of pathogen during infection spreading. This condition can be easily formulated, because the total amount of pathogen that comes out from $i$ is equal to $x_i\sum_j b_{ji}$, according to previous definitions. On the other hand, a neighbor $j$ receives from node $i$ an amount of pathogen equal to $a_{ji}x_i$, so the total quantity received by neighbors is $x_i\sum_j a_{ji}$, just by summing over the neighbors of node $i$. Thus, the conditions for no pathogen losses during transport is
\begin{equation}\label{eq:cons}
\sum_{j=1}^n a_{ji} = \sum_{j=1}^n b_{ji}
\end{equation}
for $i=1,\dots,n$. If these conditions are satisfied, we say that the model is ``conservative'', because there are no losses of pathogen during transfers between individuals. It is important to observe that the system is still open and the total pathogen's abundance in the system can increase from its level at time $t=0$ until reaching the steady state ---this is due to pathogen's reproduction inside each individual.

For our model, we have defined three plant states (node compartments) depending on host's epidemic concentrations $x_i$: individuals can be healthy (pathogen concentration equal to zero), infected (pathogen abundance above zero and below a certain death threshold $\delta$), and dead (pathogen concentration above the death threshold). Therefore, the death threshold can be interpreted as the maximum pathogen amount an individual can hold to remain alive. When the death threshold is exceeded, then the plant dies and disappears from the network. At this moment, it is assumed that the complete concentration of pathogen vanishes. Thus, plant death modifies immediately network topology as well as pathogen dynamics.

\section{Qualitative analysis without mortality}
\label{sec:qua}

Before reporting simulation results obtained for different network architectures, here we briefly summarize the qualitative analysis of Eq.~\eqref{eq:dyn} in the absence of plant mortality (i.e, in the limit $\delta\to\infty$) in two   scenarios besides the conservative case defined above:
\begin{itemize}
\item[(i)] A \emph{per-capita flux balance} condition is satisfied. In this case, the overall in-going per-capita flux of a node is balanced by the overall out-going per-capita flux associated to the same node. This condition reduces, for an arbitrary node $i$, to
\begin{equation}\label{eq:balan}
\sum_{j=1}^n a_{ij} = \sum_{j=1}^n b_{ji},\qquad i=1,\dots,n.
\end{equation}
Notice the difference between~\eqref{eq:cons} and~\eqref{eq:balan}.
\item[(ii)] \emph{No restriction} is imposed in per-capita flux matrices $A$ and $B$. 
\end{itemize}
We do not make here explicit assumptions about network structure, which is implicitly contained in matrices $A$ and $B$.

Case (i) has been contemplated because it is amenable to a complete stability analysis. The proofs of the following results are provided in Appendix~\ref{sec:appA}. It can be shown that, if the network is connected, our model always exhibis two equilibrium points: $\bm{x}^{\star}=\bm{0}:=(0,\dots,0)$ and $\bm{x}^{\star}=K\bm{1}=(K,\dots,K)$, for $\bm{1}:=(1,\dots,1)$. Moreover, the first equilibrium point, which corresponds to pathogen's clearance, is unstable. The second one, associated to a full infection situation (all individuals are infected), is globally asymptotically stable, though. Therefore, in the absence of a death threshold, it is expected that all individuals within a connected component of the network will end up infected.

The stability analysis yields comparable results for the general case of arbitrary (unrestricted) per-capita pathogen fluxes between individuals. In (ii) it is not possible to compute explicit expressions for the equilibrium points. However, if matrices $A$ and $B$ are random, it is \emph{almost sure} that, for connected networks (i.e., there is a single giant component), the only equilibria observed for Eq.~\eqref{eq:dyn} correspond to either pathogen's clearance, $\bm{x}^{\star}=\bm{0}$, or full infection, $\bm{x}^{\star}=(x_i^{\star})$, with $x_i^{\star}>0$ for every node, $x_i^{\star}$ not necessarily equal. Two situations can occur, depending on network's connectivity: (a) a single equilibrium point emerges, $\bm{x}^{\star}=\bm{0}$, which is stable; and (b) both equilibrium points arise, being unstable the one associated to pathogen's clearance, and stable the one relative to pathogen's full infection. Note that the full infection equilibrium emerges and is stable when the pathogen's clearance one becomes unstable.

Although it is difficult to analytically show that the full infection equilibrium exists (case b), we numerically found that this is precisely the most probable case, especially for sparse random networks above (but close to) the percolation threshold. 
For larger network average connectivity, case (a) becomes more frequent (not dominant, though), and model realizations in which infection disappears can be observed in simulations (see Fig.~\ref{fig:B2} in Appendix~\ref{sec:appB}). Observe that the qualitative analysis summarized here for the general case (ii) applies as well for the conservative scenario defined by Eq.~\eqref{eq:cons}.

Therefore, in the absence of a death threshold, it is expected that the pathogen spreads throughout the entire network with high probability if the network is connected. Considering a maximum value in pathogen's concentration above which an individual is regarded as dead, however, makes very hard to predict which plant is to become extinct and which one will survive (although infected) starting from arbitrary initial conditions, unless resorting to the numerical integration of the system of differential equations.

\section{Plant mortality and network dynamics}

In the presence of a (finite) death threshold $\delta$, every infected plant whose epidemic charge exceeds the mortality threshold will die. When mortality occurs, automatically, the corresponding node is disconnected from the network, as well as the associated in- and out-going fluxes (which are set equal to zero). Importantly, network topology is temporally coupled with the infection propagation dynamics. Initial networks initially can change its size and topological properties due to the demise of some of the individuals during the epidemic. As we will show below, the introduction of a mortality threshold in pathogen's concentrations changes drastically the scenario portrayed in the qualitative analysis of the model without such an upper bound on tolerable epidemic charges.

In simulations, the entries of matrices $A$ and $B$ were independently and randomly drawn from a uniform distribution $U(0,1)$. The dynamical model~\eqref{eq:dyn} was integrated numerically until convergence to an equilibrium steady state. As the dynamics unfolds, a number of nodes can go extinct until no concentration exceeds the threshold across the system. Nodes that remain alive will reach the corresponding equilibrium state.

Dynamics was integrated using two network architectures: random graphs drawn from the Erd\"os-Renyi (ER)~\citep{erdos1959}, and random geometric (RG)~\citep{antonioni2012} models. In the absence of any particular knowledge about plant locations, we assume that links are drawn at random. Specifically, it is assumed that there is a uniform probability of linking two nodes, as defined by the classical ER model. The outcome of the ER model is a graph $G(n,p)$ in which $n$ is the number of nodes and $p$ is the probability that two nodes are linked ---i.e., links are placed randomly and independently between distinct pairs of nodes with probability $p$.

In reality, however, the probability of contagion is straightforwardly related to distance. As a consequence, it is natural to assume that closer nodes have a larger probability of being connected. This assumption is applied to build RG networks. A random geometric network is the simplest kind of spatial network. Nodes are embedded in a metric space, and two nodes could be connected if and only if the euclidean distance between them is in a given range ---for instance, smaller than a certain neighborhood radius, $R$~\cite{penrose2003}. Therefore, two nodes that are within the same area of influence are randomly connected as in the ER model. Nodes that are out of this area of influence are not connected.

Nodes in our RG networks are drawn uniformly on the unit square $[0,1]\times [0,1]$. We assume that connection probability explicitly depends on the distance between nodes. We used the following continuous function,
\begin{equation}\label{eq:RGprob}
p_{ij} = 
\begin{cases}
1-\left(\frac{d_{ij}}{R}\right)^2, & d_{ij}\le R,\\
0, & d_{ij} > R,
\end{cases}
\end{equation}
which means that when the distance between $i$ and $j$ is smaller than $R$, there is a non-zero probability of connection dependent on that distance (the closer the nodes are, the more chances they are connected). When the distance between pairs of nodes is above the radius, they are never connected. Samples of the RG network model with size $n$ and radius $R$ will be denoted by $G(n,R)$. Our results are not dependent of the specific functional form of connection probability given by Eq.~\eqref{eq:RGprob}.

\begin{figure*}[t!]
\centering
\includegraphics[width=1\linewidth]{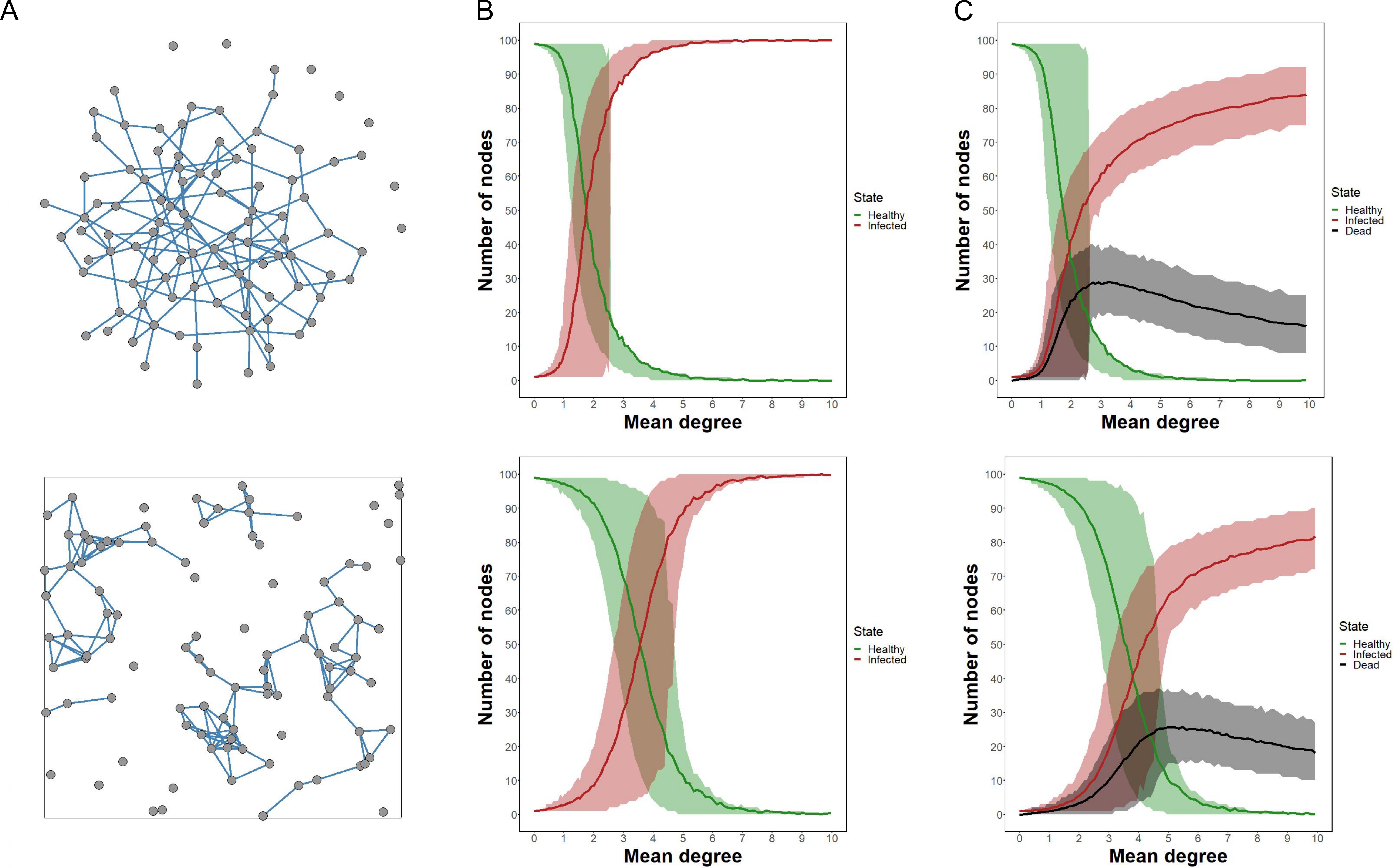}
\caption{\textbf{Number of healthy, infected and dead plants.} Upper (lower) panels correspond to ER (RG) graphs. \textbf{A}, upper panel: example of an ER network drawn from $G(100,0.03)$ with a connection probability $p=0.03$; lower panel: example of a RG network with $n=100$ and $R=0.15$ on the unit square. \textbf{B}, number of healthy (green) and infected (red) nodes in simulations carried out without a death threshold ($\delta\to\infty$) as functions of mean degrees of networks. \textbf{C}, number of healthy (green), infected (red), and dead (black) nodes in the presence of the death threshold $(\delta = 0.9K)$ as functions of initial network's mean degree. In \textbf{B} and \textbf{C} panels, the horizontal axis is calculated as the expected mean degree of the initial network ---i.e, the mean averaged degree over realized networks, which is equal to $\text{E}[k]=(n-1)p$ for ER sampled graphs. The shading areas are defined by the $10$-th and $90$-th percentiles. Here $r=0.5$ and $K=1$. 
A maximum in the number of dead nodes arises in both cases; in ER networks this happens at a mean degree of around $3$, and in RG networks at a mean degree of around $5$. Each panel was calculated by averaging over $1200$ initial conditions, networks and matrix realizations.}
\label{fig:fig1}
\end{figure*}

Simulations were conducted using networks formed by $100$ nodes. Unless the contrary is specified, we used a growth rate value $r=0.5$ and a carrying capacity value $K=1$. As for the death threshold, we set it as $\delta = 0.9K$. As initial condition, we picked up randomly a node as infected with an initial pathogen's load randomly drawn from the uniform distribution $U(0,\delta)$.

\section{Simulation results}

In Figure~\ref{fig:fig1} we report results for the infection process operating on networks with increasing connectivity. Panel A shows two samples of ER and RG network models with about $L=150$ links overall. Clearly, for that number of links, ER graphs are closer to the percolation threshold ---estimated as $p_c\approx 1/n$, \citep{estrada2012structure}---, which is evidenced by a large connected component. For $L \approx 150$, RG networks are comprised of more isolated nodes and smaller components. The radius above which a giant component arises in RG graphs is estimated as $R_c\approx \sqrt{\log n/(n\pi)}$ \citep{penrose2003}. Panel B shows how the number of healthy and infected nodes vary as functions of the mean degree when plant mortality is not considered (i.e., when the death threshold goes to infinity), whereas panel C reports the same results for the number of healthy, infected, and dead nodes versus mean degree in the presence of a finite death threshold. We observe in both panels that ER and RG curves are similar, but displaced to the right for RG graphs. 

As shown in Appendix~\ref{sec:appB}, when the network is connected, we expect that every node is infected (at least in the range of average degrees reported in Fig.~\ref{fig:fig1}), in the absence of plant mortality. For mean degrees well above the percolation threshold, this is what we observe in panel B. On the other hand, if the network is not connected, it is apparent that pathogen's spreading will not progress in those components that were not initially infected. It is only close to the percolation threshold of these networks that the curve of healthy individuals starts declining. This is consistent with the expected degree $\text{E}[k]$ at the percolation threshold ($k_i$ stands for the degree of node $i$). For an ER graph $G(n,p)$, $\text{E}[k]\approx (n-1)p$ \citep{estrada2012structure}, which is of the order of unity at $p=p_c\approx 1/n$. For a RG graph $G(n,R)$, $\text{E}[k]\approx n\pi R^2$~\citep{penrose2003}, which reduces to $\text{E}[k]\approx \log n$ at $R=R_c \approx \sqrt{\log n/(n\pi)}$. Therefore, for $n=100$ we expect the transition at mean degrees about $4.6$ for RG networks and about $1$ for ER networks. The decline of the number of healthy plants in both cases is observed at mean degrees consistent with these estimations (Fig.~\ref{fig:fig1}). This analysis suggest that network connectivity is an important driver of the overall outcome of infection spreading.

Plant mortality arises for finite values of the death threshold $\delta$. According to Fig.~\ref{fig:fig1}, the number of dead nodes peaks at a mean degree value around $\text{E}[k]\approx 3$ for ER graphs and around $\text{E}[k]\approx 5$ for RG networks. We observe this maximum in mortality independently of the values taken by the growth rate, carrying capacity, and death threshold. These parameters values do affect the location of the maximum, though. The same phenomenon occurs in the conservative scenario for ER graphs, see Figure~\ref{fig:fig1a}. Conservation of transported pathogen in RG networks yields similar results.

\begin{figure}[t!]
\centering
\includegraphics[width=\linewidth]{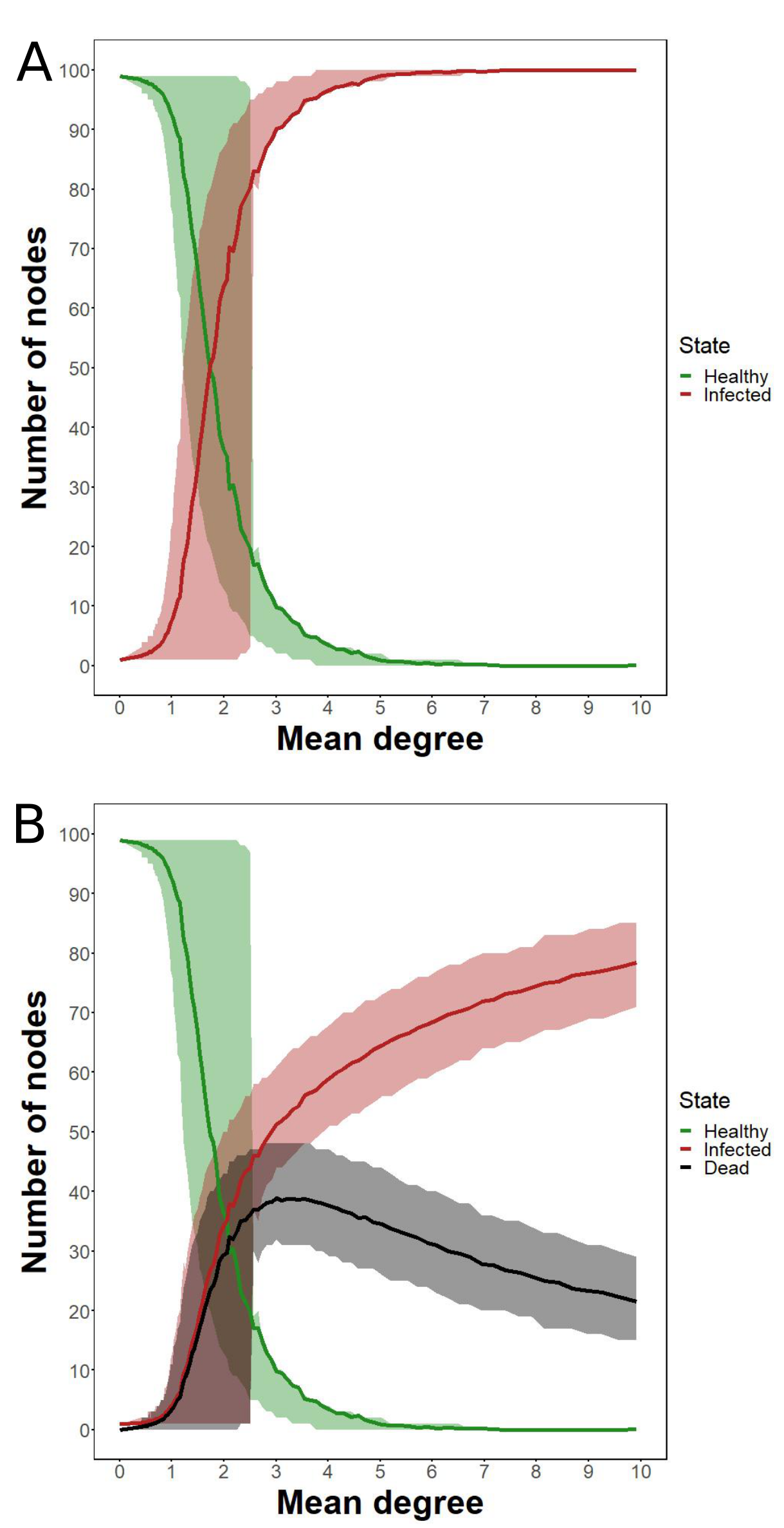}
\caption{\textbf{Conservative model.} Here we reproduce panels \textbf{B} and \textbf{C} of Fig.~\ref{fig:fig1} for the dependence of the number of individuals within compartments (healthy, infected, or dead) with ER network's mean degree, when conservation of pathogen ---Eq.~\eqref{eq:cons}--- is imposed in the dynamics.}
\label{fig:fig1a}
\end{figure}

These maxima in the number of dead plants seem counter-intuitive. One might expect that the higher the connectivity of the network, the lower the number of dead plants. However, there is an intermediate network connectivity with the highest death rate.
To test whether these maxima were due to the networks being disconnected for lower mean degrees, we repeated these simulations for small-world networks (Appendix~\ref{sec:appC}). We used the Watts-Strogatz (WS) small-world model~\citep{Watts1998}, that introduces link rewiring of regular networks yielding graphs formed by a single connected component. In this case, we do not observe a maximum in the number of dead nodes for any mean degree; plant mortality decreases monotonically as the mean degree increases. The reason is that, in small-world networks, no matter how low the mean degree is, the network is always connected. Infection can spread throughout the network, and for larger degrees the epidemic load can be channeled out through a higher number of ways, thus reducing the asymptotic pathogen's abundance of each node and, consequently, not exceeding the death threshold.

Once the network is connected, or has just a few isolated nodes, if nodes are connected on average with a higher number of neighbors, pathogen particles can be transported more efficiently between plants and the overall pathogen's charge is distributed across the network, leading to more individuals with concentrations below the threshold $\delta$. It is difficult for the pathogen to reach the death threshold in a given individual, for there are many out-going fluxes to many other individuals, swiftly distributing the pathogen.

This hypothesis can be confirmed in Figure~\ref{fig:fig2}, which shows the relationship between network connectivity and pathogen's temporal dynamics. We have considered ER graphs with three connectivity values: (i) below the maximum in the number of death nodes ($\text{E}[k]=1$), (ii) at the maximum ($\text{E}[k]=3$), and (iii) well above the maximum ($\text{E}[k]=8$). For these cases we have obtained the distributions of the size of the connected component which the initially infected node belongs to, as well as the size distribution of the largest connected component that ends up fully infected ---note that, according to our analytical results, once the process has relaxed and every survivor has abundance below the death threshold in an infected component, we expect that the full component will remain infected until reaching the steady state (this is illustrated in Fig.~\ref{fig:fig2}, lower panels). As the initially infected node is in a larger component, we observe that the temporal dynamics leads to increasing death events. However, for mean degrees leading to fully connected networks, the average pathogen's load per individual is smaller, because pathogen particles are more evenly distributed due to network's higher connectivity. Indeed, for low mean degree ($\text{E}[k]=1$), the size of the initially infected component is small, and the network is initially broken into pieces, so infection is able to progress only within the initial group and the expected number of deaths is low. For intermediate mean degree ($\text{E}[k]=3$), the initially infected component is large, but the epidemic process breaks the network and, as a consequence, the size of the largest component that ends up fully infected lowers. For large connectivity ($\text{E}[k]=8$), almost every plant belongs to the initially infected component. Such connectivity allows for an effective distribution of the overall pathogen's load across nodes, many of them remaining below the death threshold. The initial component's size in this case decreases because dead nodes are removed from the initial cluster. Hence, the maxima in mortality are explained: once a network is connected, higher connectivity implies lower mortality rates, due to the distribution of pathogen over a larger number of nodes as the epidemic spreads.

\begin{figure*}[t!]
\centering
\includegraphics[width=0.92\linewidth]{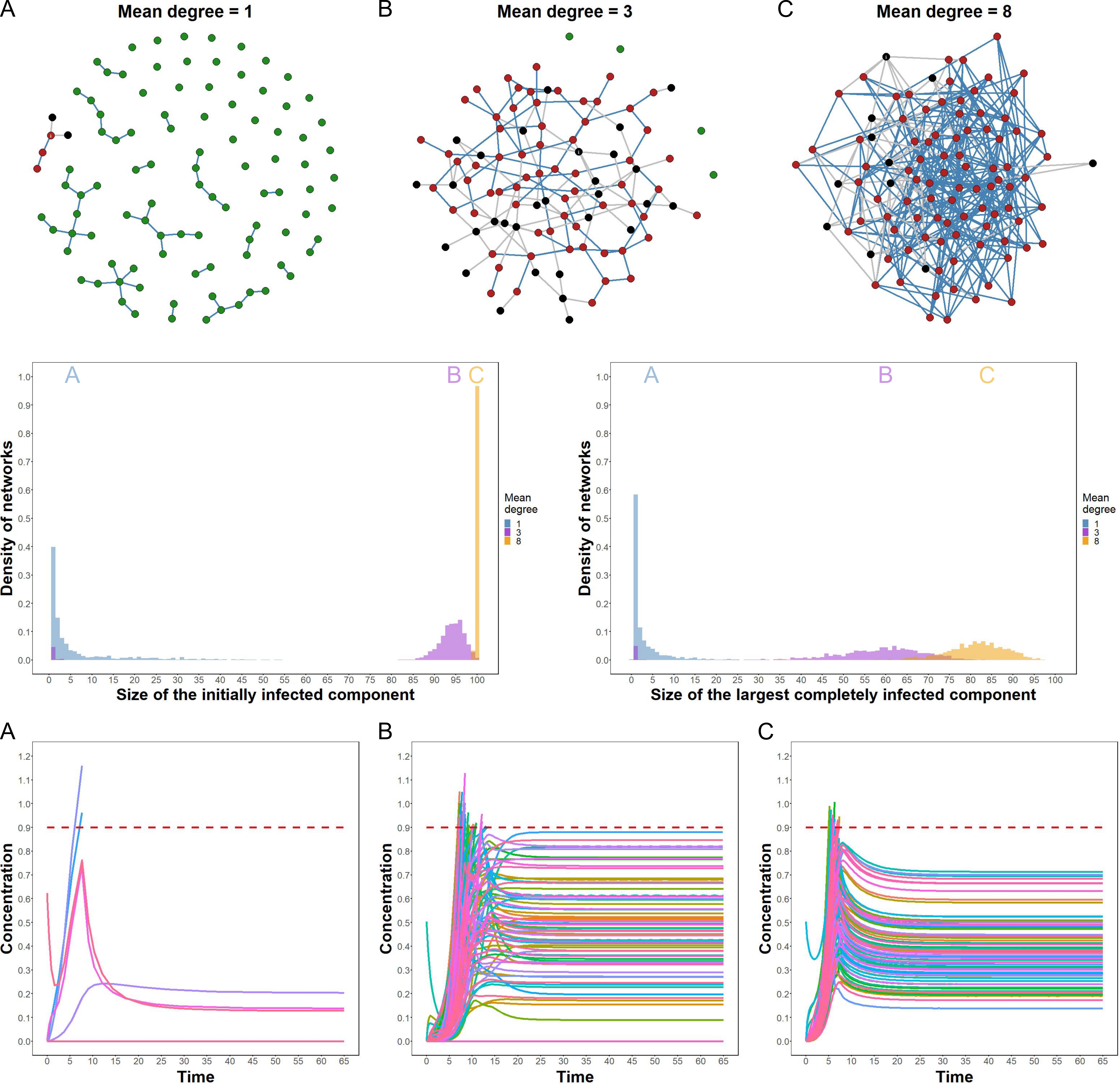}
\caption{\textbf{Random network connectivity and pathogen's temporal dynamics.} Panels \textbf{A}, \textbf{B}, and \textbf{C} correspond to ER networks with mean degrees $\text{E}[k]=1$, $\text{E}[k]=3$, and $\text{E}[k]=8$, respectively. Upper panels depict the endstate of sampled networks after the disease has spread. Healthy (green), infected (red), and dead (black) nodes are represented. Blue links remain after reaching the endstate (links connecting dead nodes are colored in grey). In all cases, in the initial state of the network only one node is infected. As for the endstate, in \textbf{A} the fractions $(H,I,D)$ of healthy, infected, and dead plants are $(H,I,D)=(0.95,0.03,0.02)$. In \textbf{B} and \textbf{C} we end up with fractions $(H,I,D)=(0.04,0.66,0.3)$ and $(H,I,D)=(0,0.89,0.11)$, respectively. The left middle panel shows the distribution of the size of the connected component which the initially infected node belongs to ($1200$ realizations for each mean degree). Blue, purple, and orange histograms correspond to $\text{E}[k]=1$, $\text{E}[k]=3$, and $\text{E}[k]=8$, respectively. The right middle panel depicts the size distribution of the largest completely infected component at the end of simulations, showing that networks break up once the dynamics unfolds. Lower panels show the temporal variation of each node pathogen's concentration. The dashed line indicates the death threshold value. The location of the distributions, which correlates to existence of a giant component, is put into correspondence with temporal dynamics, leading to a small fraction of infections and deaths (\textbf{A}), a maximum number of dead nodes (\textbf{B}), and the distribution of pathogen's loads across the network, leading to fewer deaths (\textbf{C}).}
\label{fig:fig2}
\end{figure*}

To see how remaining model parameters affect plant mortality, we set the carrying capacity to $K=1$ and varied the death threshold ($\delta \in [0.1,1]$) and the growth rate ($r\in [0.1,5]$). The percentage of dead nodes as function of the growth rate and the death threshold are shown in the heat map represented in Figure~\ref{fig:fig3}.

\begin{figure*}[t!]
\centering
\includegraphics[width=0.9\linewidth]{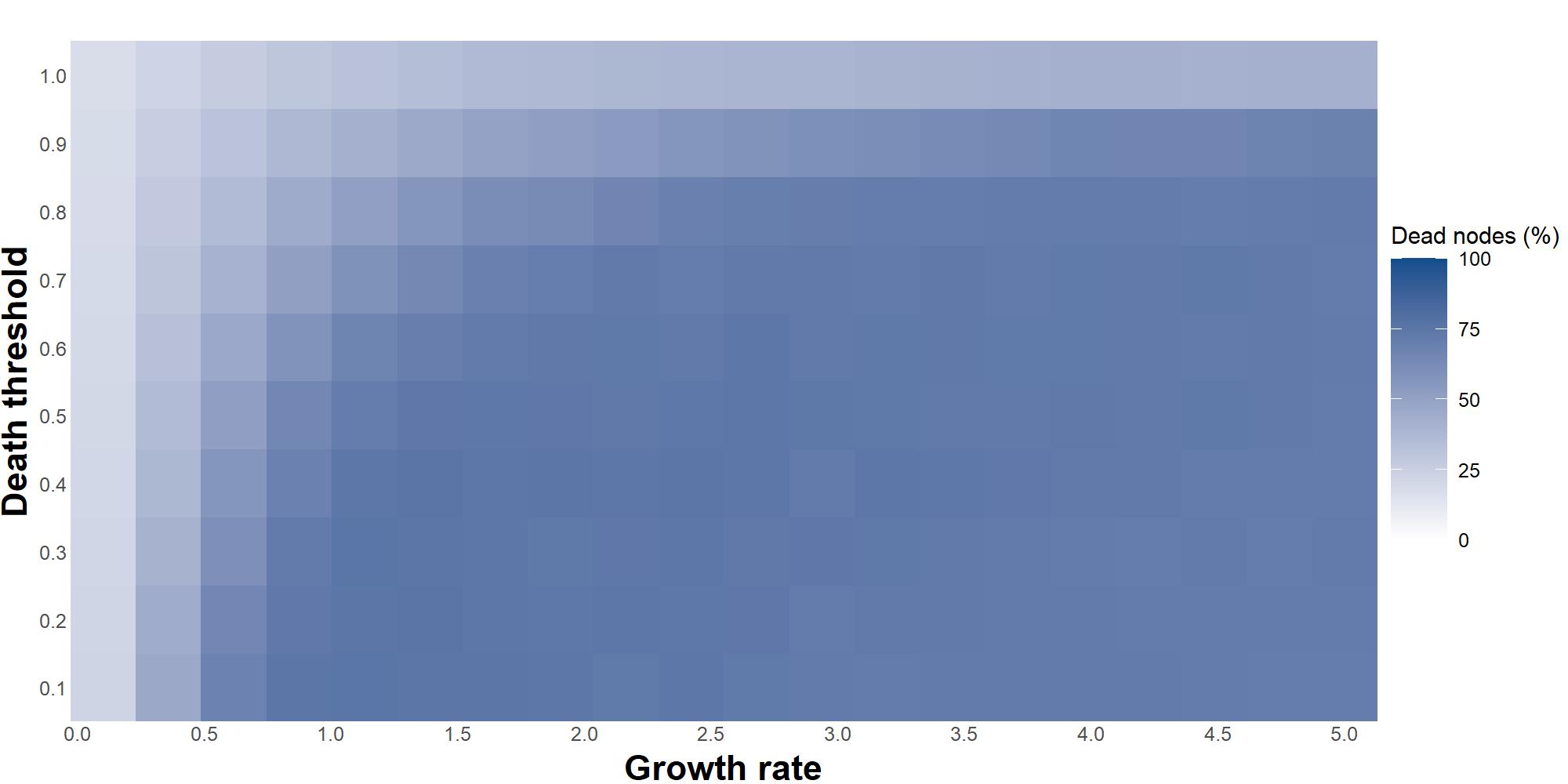}
\caption{\textbf{Fraction of dead nodes for variable death threshold and growth rate.} For fixed growth rates, mortality decreases as the death threshold increases. When the death threshold is constant, the fraction of dead nodes augments for increasing growth rates until reaching a rather constant plateau. This can occur sharply or smoothly depending on the values of the death threshold.}
\label{fig:fig3}
\end{figure*}

For fixed values of the growth rate, the number of dead nodes decreases as the death threshold increases. This is an intuitive result, since a higher death threshold implies a lesser fraction of killed plants. When the death threshold is equal to the carrying capacity ($K=1$), the number of death nodes for any growth rate is significantly lower than with any other death threshold.

If the death threshold is fixed, mortality increases as the growth rate augments until reaching an approximately constant value that depends on $\delta$. Such increase can be sharper (smoother) for smaller (larger) values of $\delta$. This is due to the fact that, when the ratio between growth rate and death threshold is large enough, nodes die too fast, disconnecting the network and hindering the spread of the disease. In summary, mortality becomes more pronounced for large pathogen's growth rates and small values of the maximum pathogen's concentration that an individual can tolerate.

\section{Discussion}

In this work we have presented and fully analyzed a mathematical model of pest dynamics, acting on plant communities, that keeps track of pathogen's concentration across individuals over time. Our approach considers a death threshold, above which plants die. We have shown that the introduction of this threshold modifies substantially the infection dynamics. Contrary to the case where no threshold exists, individual deaths modify network topology over time, because individuals whose pathogenic charge exceeds the threshold are no longer connected. A counter-intuitive result turns out: networks with higher connectivity yield a larger fraction of individuals that survive. We can explain the phenomenon using a channeled flux analog: if the percolation threshold has been crossed over and there exists a giant component, for equal pathogen growth rate and death threshold values, a higher connectivity enables a better drainage of pathogen particles across the network, which precludes a larger number of individuals from reaching their death threshold. What remains is a network formed by infected individuals that are resistant to the pathogen in the long run.

As for the robustness of our main result, we have used different network structures (ER, RG) and implemented different conditions referred to potential pathogen's losses or gains during transportation (balanced fluxes, conservative approach, unrestricted fluxes). The maximum in mortality arises as long as the network changes from being disconnected to connected ---we do not find the maximum in model networks that have no percolation transition, such as in the WS model. Therefore, the incidence of the epidemic across a population is determined by network structure. As far as we know, fatality rates in plant epidemics have not been previously connected with network average connectivity. In addition, we varied model parameters (pathogen's growth rate and death threshold) and ---not surprisingly--- we found that restrictive death thresholds and highly reproductive pathogens yield the highest fractions of dead nodes. 

Finally, we have used random values for per-capita fluxes of pathogen between individuals. Random matrices represent an appropriate framework to model large systems for which it is difficult to infer actual interaction strengths, as well as to provide consistent predictions about diversity and stability~\citep{servan2018coexistence,allesina2012stability}.

Except for the position of the maximum, which in RG networks takes place at larger mean degrees, we did not find remarkable differences with the dynamics on ER networks regarding fatality and infectivity. The RG model is more realistic because nodes are embedded in a plane, as plants in communities, and neighboring relations between them are better defined depending on distance. However, these features are not determinant when it comes to quantify mortality rates during epidemic spreading.

Temporal dynamics shown in Figure~\ref{fig:fig2} illustrate the appearance of the maximum in mortality. When a giant component starts forming, most of the nodes get infected, resulting in the propagation of the disease. However, since the mean degree is not very high, it takes longer to reach all the nodes (see Fig.~\ref{fig:fig2}B, lower panel). Pathogen grows inside infected hosts and take longer to spread, so a higher number of plants can effectively cross over the death threshold, increasing mortality numbers. Larger connectivity (above the maximum), however, implies that dead nodes reach the threshold in a more synchronous way (Fig.~\ref{fig:fig2}C, lower panel). Because connectivity is large, the network remains connected and the qualitative analysis provided in Section~\ref{sec:qua} ensures that a full infection equilibrium state (for the plants that remain alive) will be reached with high probability. 

Countless studies based on compartmentalized models have focused on finding thresholds in effective spreading rates above which infection progresses and reaches every node (see~\citep{pastor2015epidemic} and references therein). For SIS models on heterogeneous networks, the mean field epidemic threshold is known to be equal to $\langle k\rangle / \langle k^2\rangle$, where $\langle \cdot \rangle$ stands for an average value~\citep{pastor2002epidemic}. Such thresholds do not arise in our model because pathogen can be effectively spread across the giant component. For that reason, our results have to be recast in terms of how pathogen distributes depending on network connectivity.

Our work helps infer effective mechanisms that can hinder, or even halt, epidemic spreading across plant communities. For communities in which plants are scattered, it is hard that the infection propagates among disconnected network components, at least, in short time scales. Obviously, long-range dispersal can occur and uncolonized components can become infected in the long run. However, in diverse communities, packed individuals are accessible to pathogen transfers. Probably local containment to refrain propagation is not effective in highly connected systems. Our results suggest that a plausible mechanism to alleviate pathogen charges across communities is precisely increasing the connectivity by planting new individuals. For example, \textit{Viburnum acerifolium}~\citep{plantDB} is a shrub that coexists with the Californian oak, and both are infected by \textit{P. ramorum}~\citep{werres2001phytophthora}. Shrub individuals grow fast compared to trees and could alleviate oak's pathogen levels shortly after being planted, if oak's epidemic spreading could be effectively approximated by the mechanisms that drive our model dynamics. New individuals will ``drain'' pathogen particles from the species to be protected. Although infection may pervade the whole system, larger connectivity values will lower pathogen's loads, which will be more tolerable for plants targeted by conservation strategies. Thanks to pathogen reduction measures like this, infected individuals, ultimately, will survive until recovery protocols are available.



\appendix

\section{Stability results for per-capita flux balance}
\label{sec:appA}

In this section we provide proofs for the results stated in Section~\ref{sec:qua} for the per-capita flux balanced dynamics. We first calculate the attractors of the dynamics.

\begin{mytheorem}\label{the:1}
Assume that the graph $M$ associated to matrices $A$ and $B$ is connected, and that per-capita fluxes are balanced, i.e., Eq.~\eqref{eq:balan} holds. Then the only equilibrium points of~\eqref{eq:dyn} are either $x^{\star}_1=\dots=x^{\star}_n=0$ or $x^{\star}_1=\dots=x^{\star}_n=K$.
\end{mytheorem}

\begin{proof}
We want to solve the non-linear system of equations
\begin{equation}\label{eq:equi0}
rx^{\star}_i\left(1-\frac{x^{\star}_i}{K}\right)+\sum_{j=1}^n a_{ij}x^{\star}_j-x^{\star}_i\sum_{j=1}^n b_{ji}=0, \end{equation}
$i=1,\dots,n$. Completing squares, we can write the system as
\begin{equation}
\left(x^{\star}_i-\frac{K}{2}\right)^2=\frac{K^2}{4}+\frac{K}{r}\bigg(\sum_{j=1}^n a_{ij}x^{\star}_j-x^{\star}_i\sum_{j=1}^n b_{ji}\bigg)=0.
\end{equation}
Define the new variables $y_i:=x^{\star}_i-\frac{K}{2}$. With these new variables, the system reduces to
\begin{equation}\label{eq:equi}
y_i^2=\frac{K^2}{4}+\frac{K}{r}\bigg(\sum_{j=1}^n a_{ij}y_j-y_i\sum_{j=1}^n b_{ji}\bigg)=0.
\end{equation}
Without loss of generality, we can sort the entries of vector $\bm{y}$ by relabelling equations and write $-\frac{K}{2}\le y_1 \le y_2 \le \dots \le y_n$. Then~\eqref{eq:equi} for $i=1$ yields
\begin{equation}\label{eq:y1}
\begin{aligned}
y_1^2&=\frac{K^2}{4}+\frac{K}{r}\bigg(\sum_{j=1}^n a_{1j}y_j-y_1\sum_{j=1}^n b_{j1}\bigg)\\
&\ge \frac{K^2}{4}+\frac{K}{r}\bigg(\sum_{j=1}^n a_{1j}-\sum_{j=1}^n b_{j1}\bigg)y_1= \frac{K^2}{4},
\end{aligned}
\end{equation}
where we have used the assumption that $y_j\ge y_1$ for $j=2,\dots,n$ and the equality of in- and out-going fluxes, cf. Eq.~\eqref{eq:balan}. From~\eqref{eq:y1} we have that either $y_1=-\frac{K}{2}$ or $y_1\ge \frac{K}{2}$. We consider these two cases separately.
\begin{itemize}
\item[(a)] Let $y_1=-\frac{K}{2}$. Since the network is connected, exists a shortest-path over the network such that node $1$ is connected to node $n$, either directly through an existing link or by a finite series of steps using intermediate nodes. Let $\{i_1,i_2,\dots,i_s\}$ denote the index sequence for the path connecting node $1$ and node $n$. Then $a_{1,i_1}>0$ and assume that $y_{i_1} > -\frac{K}{2}$. Because $\sum_{j} a_{1j}y_{j}>-\frac{K}{2}\sum_{j} a_{1j}$ (notice that the strict inequality is due the variables being sorted, which implies that $y_k > -\frac{K}{2}$ for $k>i_1$), Eq.~\eqref{eq:equi} for $i=1$ reduces to
\begin{equation}
\begin{aligned}
0&=\sum_{j=1}^n a_{1j}y_j+\frac{K}{2}\sum_{j=1}^n b_{j1} \\
&> \frac{K}{2}\bigg(-\sum_{j=1}^n a_{1j}+\sum_{j=1}^n b_{j1}\bigg)=0,
\end{aligned}
\end{equation}
which is obviously a contradiction. Therefore $y_{i_1}=-\frac{K}{2}$, which implies that $y_1=\dots=y_{i_1}=-\frac{K}{2}$.

This argument can be iterated until reaching node $n$: now consider node $i_1$, which in the path is connected to $i_2$. If $i_2>i_1$ then we apply the same argument (using that $a_{i_1,i_2}>0$ in Eq.~\eqref{eq:equi} particularized for $i=i_1$) to prove that $y_1=\dots=y_{i_2}=-\frac{K}{2}$. If $i_2<i_1$, then we iterate the procedure for the following node in the sequence, $i_3$. At the end of the path we will reach node $i_s$, connected to the endpoint at node $n$. Then Eq.~\eqref{eq:equi} for $i=i_s$, together with the fact that $a_{i_s,n}>0$, implies that $y_1=y_2=\dots=y_n=-\frac{K}{2}$ in this case. This solution reduces to $x^{\star}_1=\dots =x^{\star}_n=0$.
\item[(b)] Assume now that $y_1\ge \frac{K}{2}$. Now particularize~\eqref{eq:equi} for $i=n$. Using that $y_j\le y_n$, we can find the following upper bound for $y_n^2$:
\begin{equation}\label{eq:yn}
\begin{aligned}
y_n^2&=\frac{K^2}{4}+\frac{K}{r}\bigg(\sum_{j=1}^n a_{nj}y_j-y_n\sum_{j=1}^n b_{jn}\bigg)\\
&\le \frac{K^2}{4}+\frac{K}{r}\bigg(\sum_{j=1}^n a_{nj}-\sum_{j=1}^nb_{jn}\bigg)y_n= \frac{K^2}{4},
\end{aligned}
\end{equation}
which implies that $-\frac{K}{2}\le y_n\le \frac{K}{2}$. Since $y_1\le y_n\le\frac{K}{2}$ and $y_1\ge\frac{K}{2}$ by hypothesis, we get $y_1=y_n=\frac{K}{2}$. This obviously leads to the solution $x^{\star}_1=\dots =x^{\star}_n=K$.
\end{itemize}
This completes the proof of the theorem.
\end{proof}

Observe that this result is general for connected graphs, irrespective of the specific network structure yielded by matrix $M$. In other words, we have not made any assumption on the distribution of the links, some of which can be absent (some $a_{ij}=0$ or $b_{ij}=0$), to prove that the only equilibrium points are those associated to full infection or to the absence of the infection.

We now focus on the stability of these equilibrium points. It can be shown that the equilibrium point associated to pathogen coexistence across individuals, $\bm{x}^{\star}=K\bm{1}$, is \emph{globally asymptotically stable}. We proceed as follows: (i) first we show that $\bm{x}^{\star}=K\bm{1}$ is asymptotically stable; (ii) then we show that the other equilibrium point, $\bm{x}^{\star}=\bm{0}$, is unstable; (ii) as the origin is unstable, global stability of the coexistence equilibrium arises because the state space of feasible solutions is invariant to the dynamical system, which implies that every trajectory with positive initial conditions will converge to $\bm{x}^{\star}=K\bm{1}$.

The Jacobian matrix can be expressed as $J=D+A$, where $D$ is a diagonal matrix $D=(d_{ii})$ whose diagonal entries are given by
\begin{equation}
d_{ii}=r-\frac{2r x_i^{\star}}{K}-\sum_{j=1}^n b_{ji}.
\end{equation}
The stability of the coexistence equilibrium point follows as a corollary of the Gershgorin's circle theorem, which we reproduce here for the sake of completeness:
\begin{mytheorem}\label{the:2}
Let $S$ be a complex $n\times n$ matrix with entries $s_{ij}$. Let $R_i=\sum_{j\ne i} \vert s_{ij}\vert$ be the sum of the absolute values of the non-diagonal entries in the $i$-th row. Let $D(s_{ii},R_{i})\subseteq \mathbb {C}$ be a closed disc centered at $s_{ii}$ with radius $R_{i}$ (such a disc is called a Gershgorin disc). Then every eigenvalue of $S$ lies within at least one of the Gershgorin discs $D(m_{ii},R_{i})$.
\end{mytheorem}

Then, stability follows directly:
\begin{mycorollary}\label{cor:1}
If condition~\eqref{eq:balan} is satisfied, the equilibrium point of~\eqref{eq:dyn} associated to full infection, $x^{\star}_1=\dots=x^{\star}_n=K$, is asymptotically stable.
\end{mycorollary}
\begin{proof} 
The Jacobian matrix, in this case, reduces to $J=D+A$ with $d_{ii}=-r-\sum_j b_{ji}$. The Gershgorin's circle theorem helps show that all the eigenvalues of $J$ have strictly negative real part. Indeed, for each row of $J$, the radius of the $i$-th disc is $R_i=\sum_j a_{ij}$ because $a_{ij}\ge 0$. Therefore each disc $D\big(-r-\sum_j b_{ji},\sum_j a_{ij}\big)$ is centered on the real axis, and each one is contained in the complex semi-plane $\text{Re}\,z \le -r$, because the right-most point of the disc is the real number $-r-\sum_j b_{ji}+\sum_j a_{ij}=-r$. Thus, any eigenvalue satisfies $\text{Re}\,\lambda \le -r$ and has strictly negative real part. Hence the equilibrium point is asymptotically stable.
\end{proof}

Again, this result is independent of whether matrices $A$ and $B$ have an adjacency matrix $M$ superimposed or not. It is general as long as the condition of per-capita flux balance, Eq.~\eqref{eq:balan}, holds.

On the other hand, we can show that the equilibrium point associated to pathogen's clearance, $\bm{x}^{\star}=\bm{0}$, is unstable. We state this as a proposition:
\begin{myproposition}\label{prop:1}
If condition~\eqref{eq:balan} holds, the equilibrium point of~\eqref{eq:dyn} associated to full pathogen extinction, $x^{\star}_1=\dots=x^{\star}_n=0$, is unstable.
\end{myproposition}
\begin{proof}
The Jacobian matrix has diagonal entries given by $d_{ii}=r-\sum_j b_{ji}$. Then trivially the vector $\bm{1}^T$ is an eigenvector of $J$ with eigenvalue $r$, because $\sum_{j} J_{ij}=d_{ii}+\sum_{j}a_{ij}=r+\sum_j a_{ij} -\sum_j b_{ji}=r$. As $J$ has constant row-sums, then $\bm{1}$ is eigenvector with eigenvalue equal to each row sum. Because there is an eigenvalue with strictly positive real part, this point is unstable, as stated.
\end{proof}

To finish with the qualitative analysis, we need to rule out the possibility that trajectories do not cross the boundaries of the space of feasible solutions, $x_i=0$, $i=1,\dots,n$. But it is easy to check that any trajectory starting with initial condition $\bm{x}(0)$ in the interior of the space $\mathbb{R}^n_+=\{\bm{x}\in \mathbb{R}^n \vert x_i\ge 0, i=1,\dots, n\}$ remains in that space, i.e., the space of feasible solutions is invariant. We state this as a proposition:
\begin{myproposition}
The state space of feasible solutions associated to Eq.~\eqref{eq:dyn}, $\mathbb{R}^n_+$, is invariant.
\end{myproposition}
\begin{proof}
First consider the initial condition $\bm{x}(0)=\bm{0}$. Since $\bm{x}=\bm{0}$ is an equilibrium point, the trajectory remains in $\mathbb{R}^n_+$.

Now consider an initial condition such that $x_i(0)\ge 0$ for all $i=1,\dots,n$ and some of the initial values verify $x_j(0)>0$. Then it is easy to see that $x_i(t)\ge 0$ for all $t>0$ and  all $i=1,\dots,n$. Assume that some variable $x_k(t_a)=0$ vanishes at $t=t_a>0$. At that time it holds that
\begin{equation}
\left.\frac{dx_k}{dt}\right\vert_{t=t_a}=\sum_{s=1}^n a_{ks}x_s(t_a)\ge 0.
\end{equation}
Because the derivative is non-negative, the flux of the ODE system does not allow the trajectory to cross the axis $x_k=0$. The same holds for the remaining variables. Hence any initial condition $\bm{x}(0)\in\mathbb{R}^n_+$ yields a trajectory contained in the state space of feasible solutions.
\end{proof}

This proposition, together with corollary~\ref{cor:1} and proposition~\ref{prop:1} yields the following corollary:
\begin{mycorollary}
If condition~\eqref{eq:balan} is satisfied, the equilibrium point of~\eqref{eq:dyn} associated to full infection, $x^{\star}_1=\dots=x^{\star}_n=K$, is globally asymptotically stable.
\end{mycorollary}
\begin{proof}
This follows trivially because the state space $\mathbb{R}^n_+$ is invariant and the unique stable equilibrium point is $\bm{x}^{\star}=K\bm{1}$. Hence all trajectories will converge to $\bm{x}^{\star}=K\bm{1}$ and its basin of attraction will be $\mathbb{R}^n_+-\{\bm{0}\}$, i.e., the full state space except the unstable equilibrium point.
\end{proof}

\section{Stability results for the general case}
\label{sec:appB}

The behavior described in Appendix~\ref{sec:appA} for the flux-balanced model is recovered \emph{almost surely} if no restrictions are imposed in matrices $A$ and $B$. Regarding system's equilibria, it is always found the one associated to pathogen's clearance, $\bm{x}^{\star}=\bm{0}$. According to the results provided below in this Appendix, it is very likely that an equilibrium point associated to full infection, $\bm{x}^{\star}$ with $x_i^{\star}>0$ for $i=1,\dots,n$, exists. But, if the network is connected, no equilibria can arise that combine infected and healthy individuals \emph{almost surely}:
\begin{myproposition}\label{prop:3}
Consider the dynamics~\eqref{eq:dyn} with unrestricted, random per-capita flux matrices $A$ and $B$ (with i.i.d. entries). If the adjacency matrix $M$ defines a connected graph, almost surely no equilibrium points $\bm{x}^{\star}$ exist such some pathogen abundances are positive and some of them are exactly equal to zero.
\end{myproposition}
\begin{proof}
Assume, without loss of generality, that the first $k$ abundances of $\bm{x}^{\star}$ are equal to zero, $x_1^{\star}=\dots=x_k^{\star}=0$, and the remaining ones are positive, $x_i^{\star}>0$ for $i=k+1,\dots,n$. Then~\eqref{eq:equi0} reduces, for $i=1,\dots,k$, to
\begin{equation}\label{eq:rank}
\sum_{j=k+1}^n a_{ij}x_j^{\star} = 0.
\end{equation}
Let $A_1$ be the submatrix of $A$ formed by the columns from $j=k+1$ to $j=n$ and the rows from $i=1$ to $i=k$. Then if Eq.~\eqref{eq:rank} was true, this would imply that the random matrix $A_1$ is not full rank. But this is a contradiction because any (connected) random matrix is full rank ---see Corollary 1.2 in~\citep{feng2007rank}.

Observe that, if the graph were disconnected, then we could find equilibria with non-zero entries within one or more connected components and zero entries in other components. Eq.~\eqref{eq:rank} would not impose any restriction because matrix elements $a_{ij}=0$ between disconnected components. For example, if infection spreading starts in a node within a component, at the end this component will reach a fully infected state, whereas the remaining components will have healthy individuals.
\end{proof}

Therefore, we can expect \emph{almost surely} that the only equilibrium points of~\eqref{eq:dyn} with unrestricted, random per-capita flux matrices $A$ and $B$, are either $\bm{x}^{\star}=\bm{0}$ or $\bm{x}^{\star}$ with all entries $x_i^{\star}>0$ for $i=1,\dots,n$, if the latter is found as solution of~\eqref{eq:equi0}. This is similar to what we found analytically in Appendix~\ref{sec:appA} for the balanced case. In principle, more than one single equilibrium associated to full infection with all $x_i^{\star}>0$ could arise. In practice, the majority of realizations should exhibit only the two aforementioned equilibria, as in the balanced case. Observe that this applies to the conservative case, for it being a particular case of the unrestricted flux situation.

Can we say something about the stability of these two equilibria? The following proposition holds:
\begin{myproposition}\label{prop:4}
Let $\bm{x}^{\star}=(x_i^{\star})$ be a solution of~\eqref{eq:equi0}, and
\begin{equation}\label{eq:jac}
J(\bm{x}^{\star}) = \emph{diag}\biggl(r-\frac{2r x_i^{\star}}{K}-\sum_{j=1}^n b_{ji}\biggr)+A
\end{equation}
be the Jacobian matrix evaluated at that equilibrium point ---$\emph{diag}(\bm{u})$ stands for a diagonal matrix defined by vector $\bm{u}$. Then, if the full infection equilibrium $\bm{x}^{\star}$ ($x_i^{\star}>0$) exists, it holds that 
\begin{equation}\label{eq:Jx}
(J(\bm{x}^{\star})\bm{x}^{\star})_i=-\frac{r{x_i^{\star}}^2}{K},
\end{equation}
for $i=1,\dots,n$. Moreover, if the full infection equilibrium arises, at the full pathogen's clearance equilibrium we find that
\begin{equation}\label{eq:J0}
(J(\bm{0})\bm{x}^{\star})_i=\frac{r{x_i^{\star}}^2}{K},
\end{equation}
for $i=1,\dots,n$ and $\bm{x}^{\star}$ the full pathogen infection equilibrium point.
\end{myproposition}
\begin{proof}
It is easy to compute that
\begin{multline}
(J(\bm{x}^{\star})\bm{x}^{\star})_i=rx_i^{\star}-\frac{2r{x_i^{\star}}^2}{K}-
x_i^{\star}\sum_{j=1}^n b_{ji}+(A\bm{x}^{\star})_i\\
=-\frac{r{x_i^{\star}}^2}{K}+rx_i^{\star}-\frac{r{x_i^{\star}}^2}{K}
-x_i^{\star}\sum_{j=1}^n b_{ji}+(A\bm{x}^{\star})_i,
\end{multline}
but in the last equality all the terms except the first vanish because $\bm{x}^{\star}$ is a solution of~\eqref{eq:equi0}. Similarly,
\begin{equation}
\begin{aligned}
(J(\bm{0})\bm{x}^{\star})_i&=rx_i^{\star}-
x_i^{\star}\sum_{j=1}^n b_{ji}+(A\bm{x}^{\star})_i = \frac{r{x_i^{\star}}^2}{K}\\
&+rx_i^{\star}-\frac{r{x_i^{\star}}^2}{K}-x_i^{\star}\sum_{j=1}^n b_{ji}+(A\bm{x}^{\star})_i,
\end{aligned}
\end{equation}
and we get~\eqref{eq:J0}.
\end{proof}

\begin{figure}[t!]
\centering
\includegraphics[width=0.85\linewidth]{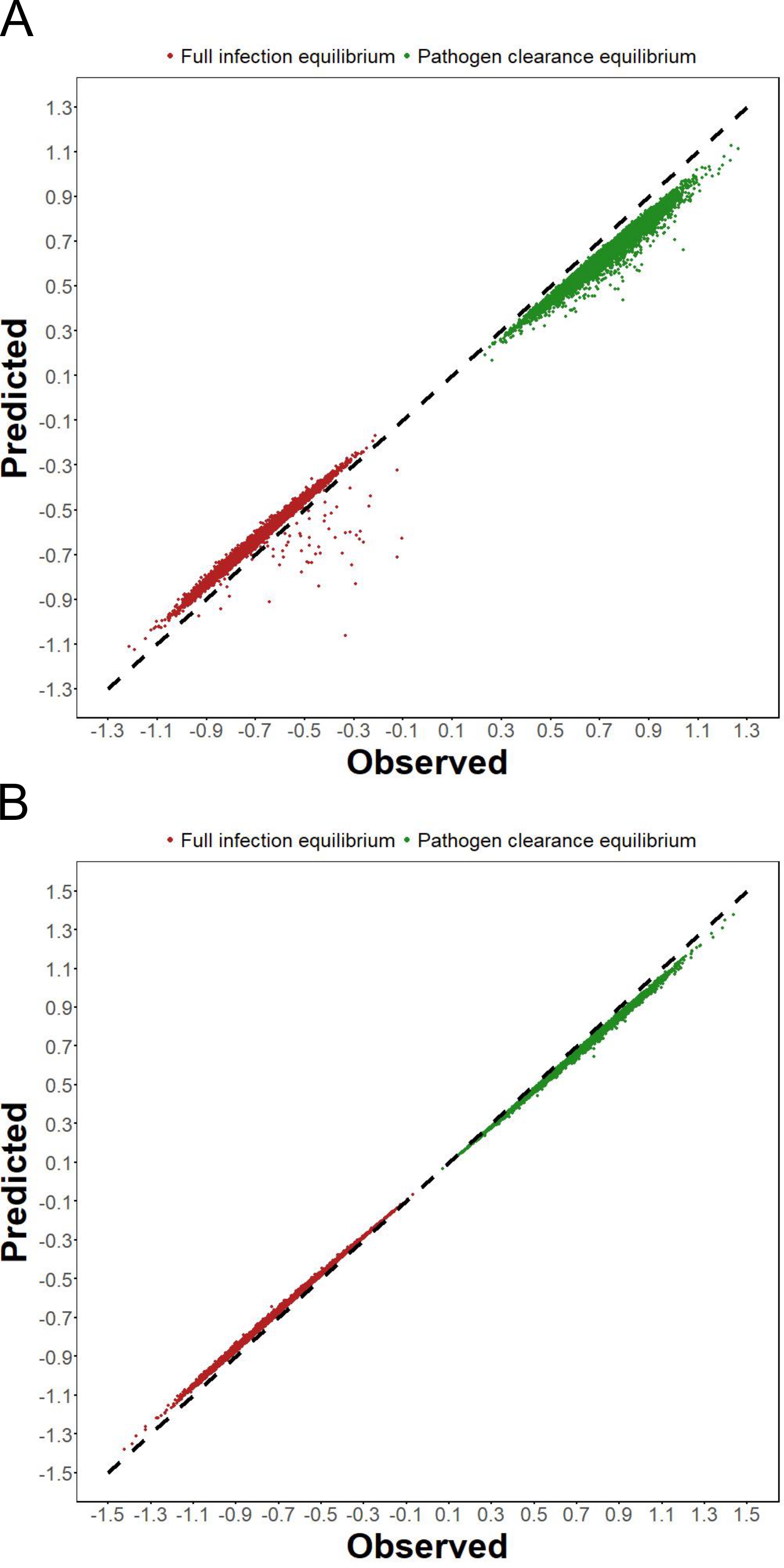}
\caption{\textbf{Predicted and observed eigenvalues.} Observed rightmost eigenvalue of the Jacobian matrix~\eqref{eq:jac} at $\bm{x}^{\star}=\bm{0}$, together with predicted values given by Eq.~\eqref{eq:J0} (green dots). Red dots stand for predicted ---cf.~Eq.~\eqref{eq:Jx}--- \emph{vs.} observed stability eigenvalues at the full infection equilibrium point, in realizations where both equilibria coexist. A total of $10^4$ Erd\"os-Renyi model networks realizations ($n=100$ nodes) were taken to produce both panels. {\bf A} corresponds to $p=0.1$, and {\bf B} was obtained for $p=0.2$, both values well above the percolation threshold. As $p$ approaches one, the agreement between predicted and observed becomes almost perfect.}
\label{fig:B1}
\end{figure}

Expressions~\eqref{eq:Jx} and~\eqref{eq:J0} can be used to provide accurate approximations of the rightmost eigenvalue of the Jacobian matrix for both equilibria. 
%
If $\bm{v}=(v_i)$ was an eigenvector of $J$, we would have that $(J(\bm{x}^{\star})\bm{v})_i = \lambda v_i$. If all entries are real and $v_i>0$, then we could write $\lambda$ as
\begin{equation}
\lambda = \frac{1}{n}\sum_{i=1}^n \frac{(J(\bm{x}^{\star})\bm{v})_i}{v_i}.
\end{equation}
Numerically we find that the eigenvector $\bm{v}$ associated with the eigenvalue that determines stability can be approximated by $\bm{x}^{\star}$, if this full infection equilibrium point exists. Then Equations~\eqref{eq:Jx} and~\eqref{eq:J0} can be written as
\begin{equation}
\frac{1}{n}\sum_{i=1}^n\frac{(J(\bm{x}^{\star})\bm{x}^{\star})_i}{x_i^{\star}}
=-\frac{r}{K}\overline{\bm{x}^{\star}}
\end{equation}
and
\begin{equation}
\frac{1}{n}\sum_{i=1}^n\frac{(J(\bm{0})\bm{x}^{\star})_i}{x_i^{\star}}
=\frac{r}{K}\overline{\bm{x}^{\star}},
\end{equation}
respectively. Here $\overline{\bm{v}}:=\frac{1}{n}\sum_{i=1}^n v_i$. If the full infection equilibrium exists and the approximation for the eigenvector $\bm{v} \approx \bm{x}^{\star}$ is correct, then the right-most eigenvalue can be approximated by $-\frac{r}{K}\overline{\bm{x}^{\star}}$ for the full infection equilibrium (which will be stable), and by $\frac{r}{K}\overline{\bm{x}^{\star}}$ for the full pathogen clearance (which, as a consequence, will be unstable). Figure~\ref{fig:B1} shows the goodness of such approximations for ER networks with different connectivity values. We observe that the approximation works better for increasing values of the connection probability $p$, yielding an almost perfect agreement in the limit $p\to 1$.
%

\begin{figure}[t!]
\centering
\includegraphics[width=\linewidth]{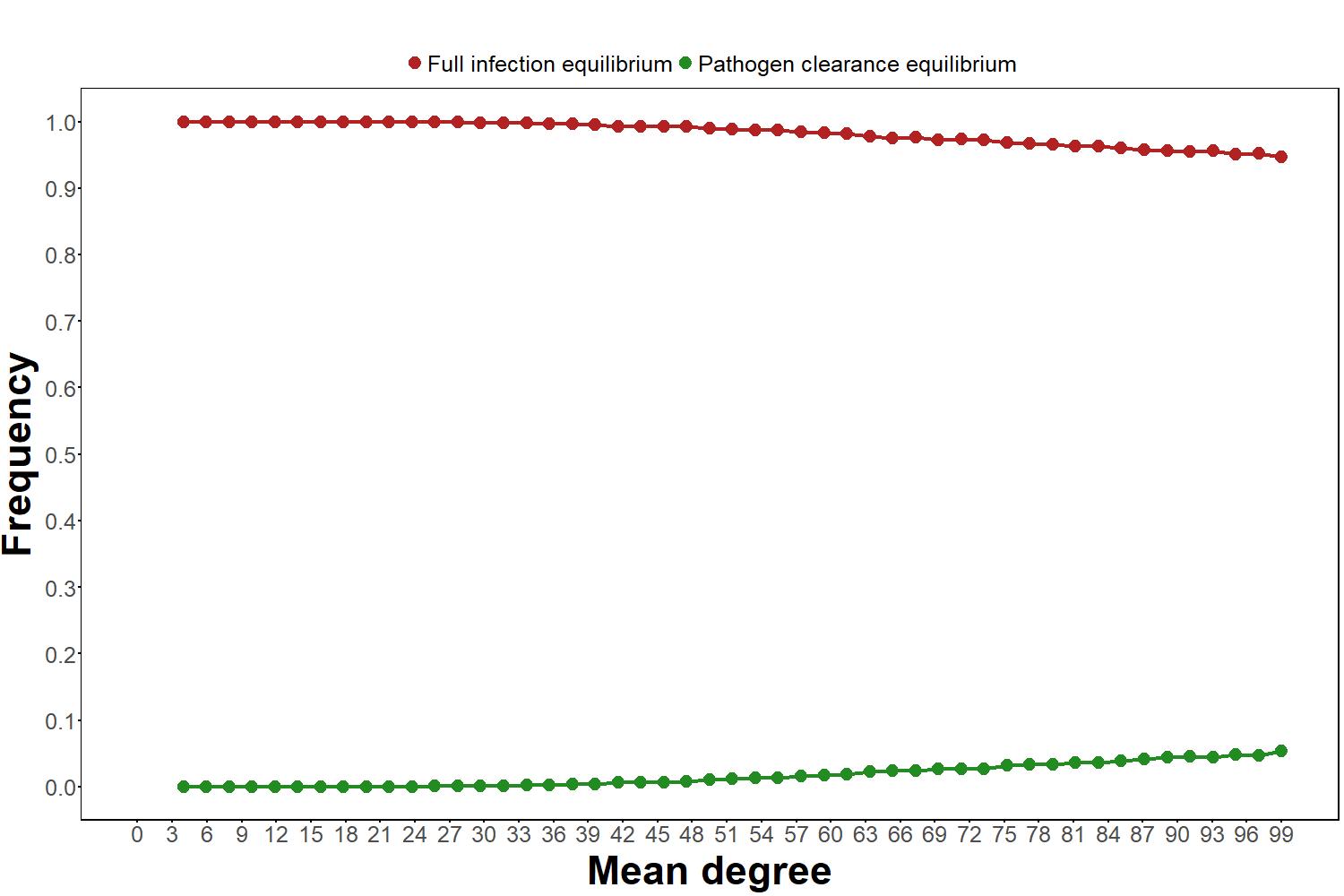}
\caption{\textbf{Probability of stability} of the full infection and pathogen's clearance equilibria as function of the expected mean degree. Ten thousand ER model realizations ($n=100$) were calculated for each connection probability $p$, and we estimated probabilities as observed frequencies for each equilibria. As $\text{E}[k]=(n-1)p$ increases, the equilibrium $\bm{x}^{\star}=\bm{0}$ becomes more likely to be asymptotically stable.}
\label{fig:B2}
\end{figure}

We conclude that, in situations where these two equilibrium points appear, the one associated to full infection is stable, and the one associated to pathogen's clearance is not. In summary, the qualitative behavior of critical points in the general case is very similar to the case of per-capita flux balance.

To finish the qualitative analysis of the general model, we have studied the effect of network connectivity in ER graphs on the stability of the two equilibria that can appear in this case. We provide the results of this exploration in Figure~\ref{fig:B2}. 

We observe that, for small mean degrees (in particular, for those reported in Fig.~\ref{fig:fig1}), the full infection equilibrium is always stable and the pathogen's clearance is unstable. However, for largely connected networks, in some cases the latter equilibrium becomes stable ---in which cases the coexistence equilibrium does not exist. This is more apparent for fully connected networks, for which about a $5\%$ of the realizations yield pathogen's clearance as endstate. It is worth mentioning that the sum of the two probabilities is numerically equal to one, so we did not find additional equilibria across all the realizations.

\section{Small-world networks}
\label{sec:appC}

In this contribution, we have shown that mortality peaks at intermediate network connectivity values, and this phenomenon can be ascribed to the potential fragmentation of the network as infection spreads. However, not every model of random network formation exhibits a percolation transition due to the formation of a giant component. According to our results, in this case the maxima should not appear.

\begin{figure}[t!]
\centering
\includegraphics[width=\linewidth]{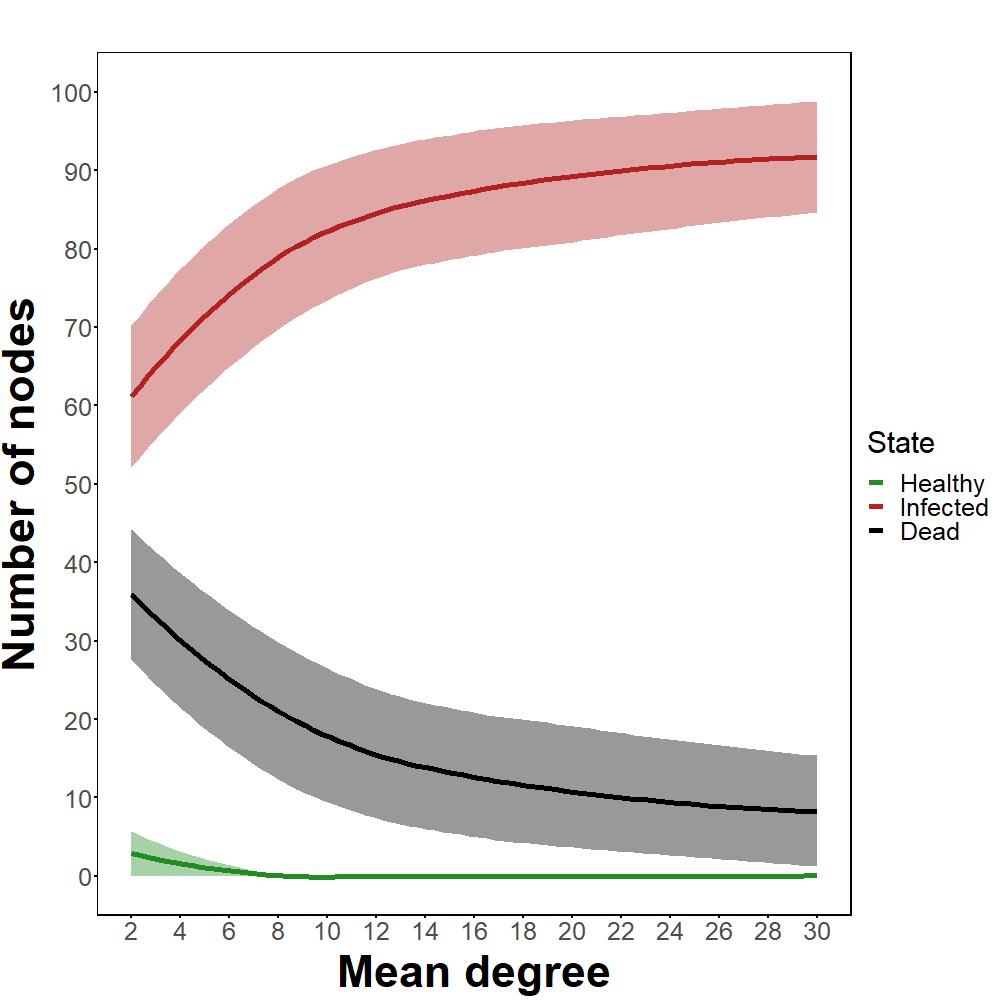}
\caption{\textbf{Small-world networks}. Here the fraction of healthy, infected, and dead nodes is represented as function of the mean degree of network samples of the WS model ---notice that the minimum possible average degree in this model is $\text{E}[k]=2$. Because model networks are connected, we observe no maximum in the number of dead nodes. We used a rewiring probability $\phi=0.01$ to generate WS model networks. Parameter values for $r$, $K$, and $\delta$ are the same as in Fig.~\ref{fig:fig1}. We averaged over $1200$ realizations.}
\label{fig:C1}
\end{figure}

Figure~\ref{fig:C1} summarizes simulation results for our epidemic spreading dynamics on to of the WS model for small-world network structure (compare these results with those reported in Fig.~\ref{fig:fig1}C). As the WS model generates connected networks, infection spreads throughout the network for any value of the mean degree. This explains that no maximum is observed in plant mortality.

\bibliography{pnas-sample}

\end{document}